\documentclass[letterpaper, 10 pt, conference]{ieeeconf}   

\IEEEoverridecommandlockouts                              
\usepackage{mathtools}
\usepackage{enumerate}
\usepackage{amsmath}
\usepackage{amssymb}
\usepackage{hyperref}
\usepackage{graphicx}
\usepackage{booktabs}
\usepackage[english]{babel}
\usepackage{xcolor, soul}
\usepackage{bm}
\usepackage{gensymb}
\usepackage{etoolbox} 
\usepackage{verbatim}

\newtheorem{remark}{Remark}
\newtheorem{corollary}{Corollary}
\newtheorem{assumption}{Assumption}
\newtheorem{definition}{Definition}
\newtheorem{proposition}{Proposition}

\newcommand\scalemath[2]{\scalebox{#1}{\mbox{\ensuremath{\displaystyle #2}}}}

\hypersetup{
    colorlinks=false,
    pdfborder={0 0 0},
}

\newbool{extended}
\setbool{extended}{true}

\newenvironment{extendedonly}{}{}
\ifboolexpr{not bool {extended}}{\AtBeginEnvironment{extendedonly}{\comment}%
\AtEndEnvironment{extendedonly}{\endcomment}}{}
\newenvironment{shortonly}{}{}
\ifbool{extended}{\AtBeginEnvironment{shortonly}{\comment}%
\AtEndEnvironment{shortonly}{\endcomment}}{}

\ifbool{extended}{\usepackage[firstpageonly=true]{draftwatermark}}{}

\title{An Offset-Free Nonlinear MPC scheme for systems \\ learned by Neural NARX models}

\author{Fabio Bonassi$^{1, *}$ \thanks{$^*$ Corresponding author}, Jing Xie$^{1}$, Marcello Farina$^{1}$, and Riccardo Scattolini$^{1}$
	\thanks{$^{1}$ The authors are with the Dipartimento di Elettronica, Informazione e Bioingegneria, Politecnico di Milano, Via Ponzio 34/5, 20133, Milano, Italy. E-mail: {\tt\small name.surname@polimi.it}}}

\begin{document}

\maketitle

\ifbool{extended}{
	\DraftwatermarkOptions{%
 		angle=0,
 		hpos=0.5\paperwidth,
 		vpos=0.95\paperheight,
 		fontsize=0.012\paperwidth,
 		color={[gray]{0.2}},
 		text={
   			\parbox{0.99\textwidth}{\copyright \, 2022 IEEE. This manuscript is the extended version of an article appearing in the Proceedings of the 2022 IEEE 61st Conference on Decision and Control (CDC), December 6-9, 2022, Cancun, Mexico, pp. 2123-2128. DOI: \href{https://doi.org/10.1109/CDC51059.2022.9992362}{10.1109/CDC51059.2022.9992362}.}},}}{}

\begin{abstract}
This paper deals with the design of nonlinear MPC controllers that provide offset-free setpoint tracking for models described by Neural Nonlinear AutoRegressive eXogenous (NNARX) networks. 
The NNARX model is identified from input-output data collected from the plant, and can be given a state-space representation with known measurable states made by past input and output variables, so that a state observer is not required. 
In the training phase, the Incremental Input-to-State Stability ($\delta$ISS) property can be forced when consistent with the behavior of the plant.
The $\delta$ISS property is then leveraged to augment the model with an explicit integral action on the output tracking error, which allows to achieve offset-free tracking capabilities to the designed control scheme.
The proposed control architecture is numerically tested on a water heating system and the achieved results are compared to those scored by another popular offset-free MPC method, showing that the proposed scheme attains remarkable performances even in presence of disturbances acting on the plant.
\end{abstract}


\section{Introduction}
With the availability of large and informative data sets and increasing computation power, learning-based methods for nonlinear system identification have become popular in the control community \cite{zhong2013learningbased, schoukens2019nonlinear}, see for example the control design algorithms based on set membership ide ntification \cite{terzi2019learning} and Koopman-based system identification \cite{korda2018linear}.

Among the most popular machine learning approaches for control, the ones relying on Recurrent Neural Networks (RNN) have been proven to provide significant results \cite{ljung2020deep}. 
Among the many model architectures proposed in the literature, it is worth mentioning here Neural NARXs (NNARX) \cite{levin1996control},  Echo State Networks (ESN) \cite{jaeger2007echo}, Long  Short Term Memory networks (LSTM) \cite{hochreiter1997long}, and Gated Recurrent Units (GRU) \cite{chung2014empirical}. 
In particular, it has been shown that these RNNs can be recast as state-space dynamical systems that can be trained to identify unknown systems provided that enough input-output measured data is available \cite{bonassi2022survey, bianchi2017recurrent}.
 
How to guarantee stability properties of these RNNs architectures, in terms of Input to State Stability (ISS) and Incremental Input-to-State stability ($\delta$ISS), has been recently studied \cite{bonassi2022survey}, see \cite{terzi2021learning} for LSTMs,  \cite{bonassi2020stability} for GRUs, \cite{armenio2019model} for ESNs, and \cite{bonassi2021nnarx} for NNARXs. 
In these works, sufficient conditions for the ISS and $\delta$ISS of RNNs are stated as nonlinear inequalities on the networks' parameters.
In \cite{bonassi2022survey} these stability properties have been shown to be useful to address the interpretability, safety and robustness issues.
Furthermore, $\delta$ISS has proven to be a fundamental tool for the design of provenly stabilizing Model Predictive Control (MPC) laws for several RNN architectures, see \cite{terzi2021learning, armenio2019model}.

A major limitation of these stabilizing MPC strategies, however, is that their static performances, i.e. their capability of steering the system's output towards a constant setpoint, is tightly related to the magnitude of the plant-model mismatch and to the presence of disturbances that affect the system.
In many applications, however, ensuring that the controller can track asymptotically-constant reference signals with zero offset might be a requirement.
In this context, several offset-free nonlinear MPC strategies have been proposed in the literature, see \cite{pannocchia2015offset} for a review on the topic.
Among them, one of the most popular is the one described in \cite{morari2012nonlinear}, in which the authors propose to augment the system model with a disturbance model, and then to design an observer to estimate its state. 
The observed state is then used to synthesize a stabilizing nonlinear MPC law.
This approach relies on the possibility to suitably model and estimate such disturbance.

An alternative approach is described in \cite{magni2001output}, which proposes to augment the system with the output tracking error integrator, and to use a state observer to reconstruct the state of such system.
Notably, this scheme can be adopted to solve the tracking problem even for time-varying references that are generated by stable exogenous systems, such as ramps or sinusoids.
This solution has been adopted in \cite{bonassi2021nonlinear} for the design of an offset-free MPC controller for systems learned by GRU networks.

In this paper, we focus on systems learned by NNARX models, which are quite popular owing to their simple structure and training.
Indeed, in NARX models the output at the future time instant is computed as a nonlinear function of  past input and output data. 
In particular, Neural NARX models are those that feature a feed-forward neural network as nonlinear regression function.
The advantage of NNARXs is that, since their state boils down to a vector of past input-output data, when these models are operated in closed-loop the state is known \cite{bonassi2021nnarx}, which makes the control design procedure significantly easier.

In this context, the goal of this paper is to design a control strategy for NNARX models that guarantees offset-free tracking of constant references, as well as the nominal stability of the closed-loop system.
Unlike the aforementioned approaches, the proposed strategy does not rely upon a state observer.
Along the lines of \cite{magni2001output, bonassi2021nonlinear}, we propose to include two elements in the control system: (\emph{i}) an output tracking error integrator, which allows to attain offset-free tracking capabilities; (\emph{ii}) a derivative action, which ensures that -- at steady state -- the regulation of the system relies entirely upon the integral action, whereas the goal of the MPC is to improve the dynamic performances and to ensure constraint satisfaction during the transient.

To provide sound guarantees, we show that the $\delta$ISS of the model, under mild assumptions, ensures that such integral action can be designed to preserve the local asymptotic stability of the closed-loop system.
The proposed approach has been tested on a water heating benchmark system, and the achieved closed-loop performances have been compared to those achieved by the strategy proposed in \cite{morari2012nonlinear}.
The simulation results show that, unlike this latter, the proposed approach attains offset-free tracking of constant references even in presence of significant plant  perturbations. \ifbool{extended}{\smallskip}{}

\begin{extendedonly}
	The paper is structured as follows. 
	In Section \ref{sec:nnarx} the NNARX models and their stability properties are presented. 
	In Section \ref{sec:control} the proposed control architecture is detailed, which is then tested on the water heating benchmark system in Section \ref{sec:example}.
	Finally, conclusions are drawn in Section \ref{sec:conclusions}.
\end{extendedonly}

\subsection{Notation}
The following notation is adopted. Given a vector $v$, we indicate by $v^\prime$ its transpose and by $\| v \|_p$ its $p$-norm.
Moreover, given a matrix $Q$, we denote $\| v \|_Q^2 = v^\prime Q v$.
For compactness, the time instant associated to time-varying vectors is reported as a subscript, e.g. $v_k$. Sequences of vectors are indicated by bold-face fonts, i.e. $\bm{v}_k = \{ v_0, ..., v_k \}$, and their $\ell_{p,q}$ norm is defined as
$\| \bm{v}_k \|_{p,q} = \big\| [ \| v_0 \|_p, ..., \| v_k \|_p ] \big\|_q$. Notably, $\| \bm{v}_k \|_{p,\infty} = \max_{t \in \{0, ..., k\}} \| v_t \|_p$.

\section{Neural NARX Model} \label{sec:nnarx}
NNARX models \cite{bonassi2021nnarx} are nonlinear, time-invariant, discrete-time models with input $u$, assumed to lie  in a compact set $\mathcal{U}\subseteq \mathcal{R}^m$, and output $y \in {R}^p$. In this paper, a square system is assumed, i.e., $p=m$.
Letting $k$ be the discrete time index, at time $k$ the future output $y_{k+1}$ is computed as a nonlinear regression function $\eta$ on past $N$ input and output samples:
\begin{equation} \label{eq:rnn:nnarx_model}
y_{k+1} = \eta(y_{k}, y_{k-1}, ..., y_{k-N+1}, u_{k}, u_{k-1}, ... u_{k-N}; \Phi),
\end{equation}
where $\Phi$ indicates the model's parameters.
It is easy to rewrite model \eqref{eq:rnn:nnarx_model} in state space form by defining, for $i \in \{1, ..., N \}$,
\begin{equation} \label{eq:rnn:nnarx_states}
	z_{i, k} = \left[\begin{array}{c}
		y_{k-N+i} \\
		u_{k-N-1+i}
	\end{array}\right],
\end{equation}
and by denoting the state vector $x_{k} = [ z_{1, k}^\prime, ..., z_{N, k}^\prime]^\prime \in \mathbb{R}^{n}$, the model can be compactly rewritten as
\begin{equation} \label{eq:nnarx:statespace}
\begin{dcases}
  x_{k+1} = {A} x_{k} + B_{u} u_{k} + B_{x} \eta(x_{k}, u_{k}; \Phi) \\
  y_{k} = C x_{k}
\end{dcases}
\end{equation}
where $A$, $B_u$, $B_x$, and $C$ are fixed matrices with known structure and elements equal to zero or one, see \cite{bonassi2021nnarx}.

In NNARX models, the regression function $\eta$ in \eqref{eq:rnn:nnarx_model} is a Feed-Forward Neural Network (FFNN), i.e. a static map of $M$ layers of neurons.
Each layer is a linear combination of its inputs, squashed by a suitable nonlinear function, named activation function. 
A compact formulation of $\eta$ is
\begin{equation}  \label{eq:model:ffnn}
		\eta(x_{k}, u_{k}) = U_0\eta_M(\eta_{M-1}(...\eta_1(x_k,u_k),u_k),u_k)+b_0
\end{equation}
where $\eta_l$ is the nonlinear relation between $l$-th and the previous layer, which can be stated as
\begin{equation}  \label{eq:model:ffnn2}
		\eta_l(\eta_{l-1},u_k) = \psi_l \big( W_l u_{k} + U_l \eta_{l-1} + b_l \big),
\end{equation}
where $\psi_l$ is a Lipschitz-continuous activation function, applied element-wise on its argument, having Lipschitz constant $L_{\psi l}$ and satisfying $\psi_l(0) = 0$. The matrices $W_l$, $U_l$ and $b_l$ are the weights of the layer, which constitute the network's parameters $\Phi = \{ U_0, b_0, \{ U_l, W_l, b_l \}_{ l = 1, ..., M} \}$.
An example of activation function is the $\tanh$ function, see \cite{bonassi2021nnarx}.
For compactness, the NNARX model is henceforth denoted as,
\begin{equation} \label{eq:nnarx:compact}
	\Sigma: \begin{dcases}
		x_{k+1} = f(x_{k},u_{k}) \\
		y_k = C x_k
	\end{dcases},
\end{equation}
where the dependency on $\Phi$ is omitted for compactness. 

The weights $\Phi$ are learned from the input-output data collected from the system during the so-called training procedure, in which the parameters that best explain the measured data are sought.
Generally, one seeks the set of weights minimizing the simulation error, i.e. the open-loop prediction error between the model and the real system.
Entering into the details of this procedure is not among the aims of this article: the interested reader is referred to \cite{bonassi2021nnarx}.

Under the assumption that the motion of the plant to be identified displays stability properties\footnote{In particular, we assume that the plant is $\delta$ISS \cite{bonassi2022survey}. 
This property can either be known a-priori, e.g. by physical arguments, or it can verified numerically on the collected data.}, as discussed in \cite{bonassi2021nnarx} it is possible to include an additional term in the training loss function that allows to learn a provenly ISS and $\delta$ISS NNARX model.
The definition of $\delta$ISS for a generic state-space nonlinear system, such as \eqref{eq:nnarx:compact}, is reported below.

\begin{definition}[$\delta$ISS] \label{def:deltaiss}
	A system is $\delta$ISS if there exist functions $\beta$ of class $\mathcal{KL}$ and $\gamma$ of class $\mathcal{K}_\infty$ such that, for any pair of initial states $x_{a,0}$ and $x_{b,0}$, and any pair of input sequences $\bm{u}_a$ and $\bm{u}_b$, where $u_{a,k}\in \mathcal{U}$ and $u_{b,k}\in \mathcal{U}$, such that
	\begin{equation} \label{eq:deltaiss:definition}
		\| x_{a,k} - x_{b,k} \|_2 \leq \beta(\| x_{a,0} - x_{b,0} \|_2, k) + \gamma(\| \bm{u}_{a,k} - \bm{u}_{b,k} \|_{2, \infty})
	\end{equation}
	for any $k \geq 0$, where $x_{*,k}$ denotes the state trajectory of the system initialized in $x_{*,0}$ and fed by the sequence $\bm{u}_{*,k}$. 
\end{definition}

Henceforth, it is assumed that the NNARX is trained according to the prescriptions detailed in \cite{bonassi2021nnarx} in order to ensure its $\delta$ISS, allowing to verify the following assumption. \smallskip

\begin{assumption} \label{ass:delta_iss}
The NNARX model \eqref{eq:nnarx:statespace} is $\delta$ISS.
\end{assumption}

\section{Controller Design} \label{sec:control}
The main goal of this paper is, given the NNARX model \eqref{eq:nnarx:compact} of the system, to propose a solution for problem of offset-free tracking of constant references.
Specifically, for some constant output reference $\bar{y}$, we want to design an MPC law which guarantees that the output error converges asymptotically to zero, i.e.
\begin{equation}\label{eq:setp}
    e_k = \bar{y} - y_k \xrightarrow[k \to \infty]{} 0.
\end{equation}

\subsection{Linearization}
To solve this problem, we will rely upon the linearization of model \eqref{eq:nnarx:compact} around an equilibrium point $(\bar{x}, \bar{u}, \bar{y})$  satisfying 
\begin{equation} \label{eq:control:equilibrium_def}
	\begin{dcases}
		\bar{x} = f(\bar{x}, \bar{u}) \\
  		\bar{y} = C \bar{x}
	\end{dcases}
\end{equation}
where $\bar{u}$ is assumed to belong to $\mathcal{U}$. 
Let us first denote by
\begin{equation} \label{eq:control:linearization}
\begin{aligned}
	A_\delta = \left. \frac{\partial f(x, u)}{\partial x}  \right\lvert_{(\bar{x}, \bar{u})},  \quad B_\delta = \left. \frac{\partial f(x, u)}{\partial u}  \right\lvert_{(\bar{x}, \bar{u})}, 
\end{aligned}
\end{equation} 
the matrices of the linearized system around the equilibrium $(\bar{x}, \bar{u}, \bar{y})$.
To characterize the stability properties of $A_\delta$, the following result is provided. \smallskip

\begin{proposition} \label{prop:asymptotic_stability}
Consider a nonlinear system in the form of \eqref{eq:nnarx:compact}.
Assume that it is $\delta$ISS in the sense specified by Definition \ref{def:deltaiss}, and that function $\beta$ admits an exponential form, i.e. that there exist constants $\rho > 0$ and $\lambda \in (0, 1)$ such that $\beta(\| x_{a,0} - x_{b,0} \|_2, k) \leq \rho \| x_{a,0} - x_{b,0} \|_2 \, \lambda^k$.
	Then, for each equilibrium $(\bar{x}, \bar{u}, \bar{y})$ satisfying \eqref{eq:control:equilibrium_def}, the matrix $A_\delta$ \eqref{eq:control:linearization} is Schur stable.
\end{proposition}
\begin{proof}
	\ifbool{extended}{See the Appendix.}{See the extended version of the paper \cite{bonassi2022extended}.}
\end{proof}
\smallskip

Let us remark that the exponential form of function $\beta$ is indeed enjoyed by $\delta$ISS RNNs, see \cite{terzi2021learning, bonassi2020stability}, and specifically \cite{bonassi2021nnarx} for NNARXs.
Moreover, let us introduce the following Assumption. \smallskip

\begin{assumption} \label{ass:linearized}
	The tuple $(A_\delta, B_\delta, C)$  is reachable, observable, and does not have invariant zeros at $z=1$.
\end{assumption}
\smallskip

Under Assumption \ref{ass:linearized}, in light of Theorem 1 in \cite{de1997stabilizing}, one can guarantee the existence of an open neighborhood of $\bar{y}$, denoted by $\Gamma(\bar{y}) \subseteq \mathcal{R}^{m}$, where, for any $\tilde{y} \in \Gamma(\bar{y})$, there exists an equilibrium $(\tilde{x}(\tilde{y}), \tilde{u}(\tilde{y}), \tilde{y})$, where
\begin{equation}
	\begin{dcases}
		\tilde{x}(\tilde{y}) = f(\tilde{x}(\tilde{y}), \tilde{u}(\tilde{y})) \\
  		\tilde{y} = C \tilde{x}(\tilde{y})
	\end{dcases}.
\end{equation}
This local result allows to conclude that it is possible to move the output reference signal in a neighborhood of the output equilibrium $\bar{y}$ and still guarantee that a feasible solution to the tracking problem exists. 

\begin{figure}[t]
	\centering
	\includegraphics[width=\ifbool{extended}{0.9}{0.75} \linewidth]{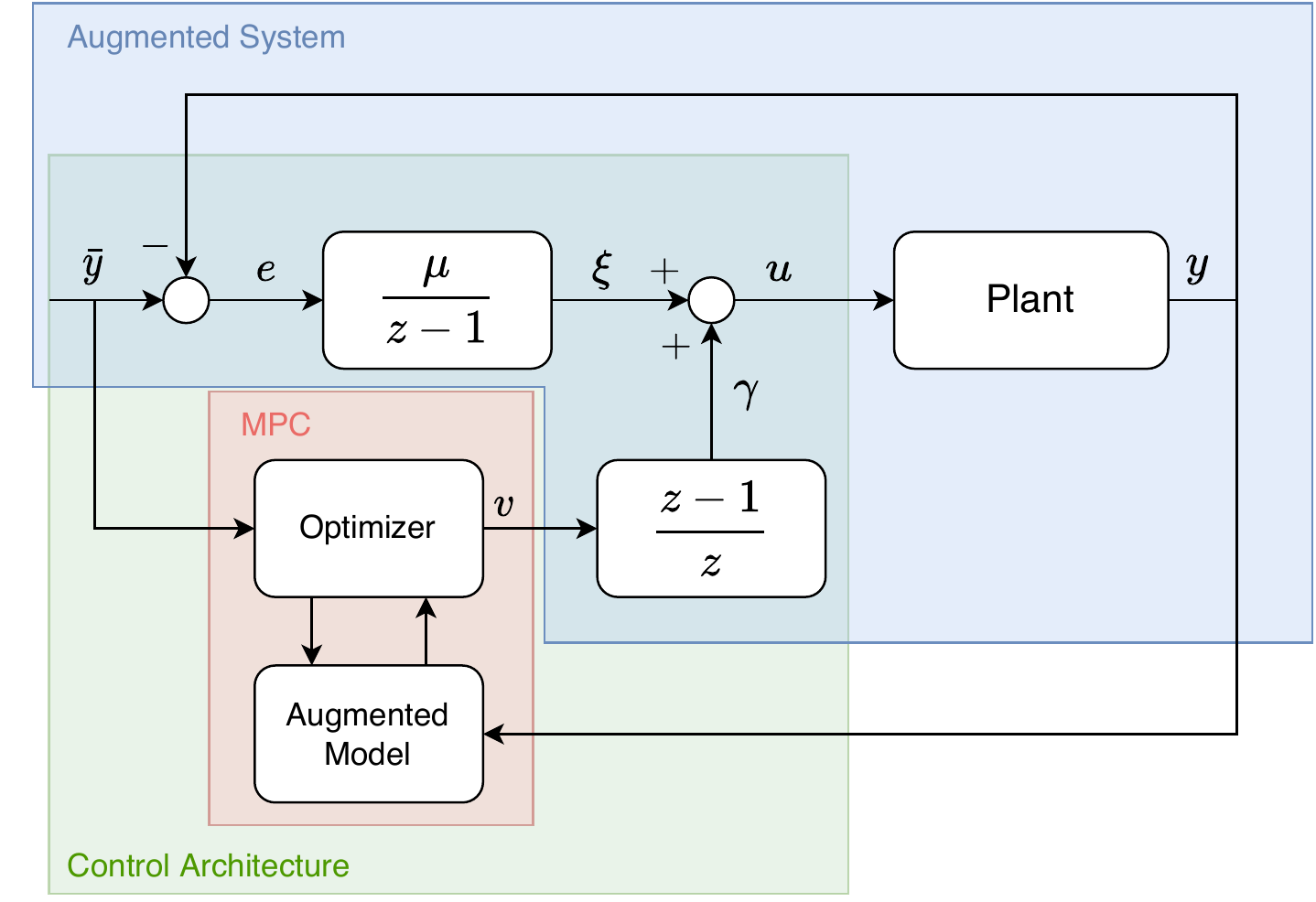}
	\caption{Schematic of the proposed control architecture}
	\label{fig:controlscheme}
\end{figure}

\subsection{The control architecture}
Once the conditions for the existence of a solution to the output tracking problem have been established, we are in the position to describe the main elements of the  adopted control architecture, depicted in Figure \ref{fig:controlscheme}, listed below
\begin{enumerate}[i.]
	\item The system is augmented with the integral of the output tracking error $e_k = \bar{y} - y_k$. Indeed, in light of the Internal Model Principle \cite{francis1976internal}, such integral action guarantees robust asymptotic zero-error regulation for constant reference signals, i.e. $e_k \xrightarrow[k \to \infty]{} 0$, and plant's parametric uncertainties, provided that the closed-loop stability guarantees are maintained. 
	\item The model is augmented with a derivative action on MPC's control variable $v$. This guarantees that, at steady state, the MPC contribution is null and the control variable entirely relies on the integral action. This approach, later detailed, is useful in the definition of suitable terminal constraints to be used in the formulation of the stabilizing MPC algorithm. The aim of the MPC regulator is that of performance enhancement and constraint handling during transients.
\end{enumerate}

As clear from Figure \ref{fig:controlscheme}, the control action $u$ is composed of two terms
\begin{equation} \label{eq:control:u}
	u_k = \xi_k + \gamma_k,
\end{equation}
where $\xi_k \in \mathbb{R}^m$ and $\gamma_k \in \mathbb{R}^m$ are the integral and derivative actions, respectively.
More specifically, the integral action is ruled by
\begin{equation} \label{eq:control:integral}
	\xi_{k+1} = \xi_k + \mu ( \bar{y} - C x_k ),
\end{equation}
where $\mu$ denotes the gain of the integral action.
Also, the derivative action $\gamma_k$ is defined as
\begin{equation} \label{eq:control:derivative}
	\begin{dcases}
		\theta_{k+1} = v_k \\ 
		\gamma_k = v_k - \theta_k
	\end{dcases}.
\end{equation}

Thus, the augmented system is obtained combining \eqref{eq:nnarx:compact}, \eqref{eq:control:u}, \eqref{eq:control:integral}, and \eqref{eq:control:derivative}, and it reads as
\begin{equation}
	\begin{dcases}
		x_{k+1} = f(x_k, u_k) \\
		\xi_{k+1} = \xi_k + \mu (\bar{y} - C x_k) \\
		\theta_{k+1} = v_k \\
		\gamma_k = v_k -\theta_k \\
		u_k = \xi_k + \gamma_k \\
		y_k = C x_k
	\end{dcases},
\end{equation}
which will henceforth be compactly denoted as
\begin{equation} \label{eq:control:enlarged}
	\Sigma_a: \begin{dcases}
		\chi_{k+1} = f_a(\chi_k, v_k, \bar{y}) \\
		\zeta_k = g_a(\chi_k)
	\end{dcases},
\end{equation}
where $\chi_k = [x_k^\prime, \xi_k^\prime, \theta_k^\prime ]^\prime$ denotes the state of the augmented system and $\zeta_k = [y_k^\prime, u_k^\prime]^\prime$ its output.

The first step of the design procedure consists of tuning the gain $\mu$, following the guidelines of \cite{scattolini1985parameter}, such that the enlarged system \eqref{eq:control:enlarged} displays stability properties. 
The following proposition can be stated. \smallskip

\begin{corollary} \label{prop:integrator}
Assume that $A_\delta$ is Schur stable, and that Assumption \ref{ass:linearized} holds.
Then, there exists $\check{\mu} > 0$ such that, for any $\tilde{\mu} \in (0, \check{\mu})$, the integrator gain
	\begin{equation}
		\mu = \tilde{\mu} \big[ C_\delta (I - A_\delta)^{-1} B_\delta \big]^{-1}
	\end{equation}
	makes the enlarged system \eqref{eq:control:enlarged}, linearized around $(\bar{\chi}, \bar{v}, \bar{\zeta})$,  asymptotically stable, where $\bar{\chi} = \big[ \bar{x}^\prime, \bar{\xi}^\prime, \bar{\theta}^\prime \big]^\prime = \big[ \bar{x}^\prime, \bar{u}^\prime, \bar{v}^\prime \big]^\prime$, $\bar{\zeta} = [\bar{y}^\prime, \bar{u}^\prime]^\prime$, and $\bar{v}$ is any constant value.
\end{corollary}
\begin{proof}
	In light of Proposition \ref{prop:asymptotic_stability}, the matrix $A_\delta$ is Schur stable. Then, thanks to Assumption \ref{ass:linearized}, the results shown in \cite{scattolini1985parameter} can be applied to prove the corollary. 
\end{proof}

\subsection{MPC design}
Having defined the augmented system model $\Sigma_a$, a stabilizing nonlinear MPC law can be designed.
Letting $\bar{\chi}$ and $\bar{\zeta}$ be the target state and output {introduced in Corollary \ref{prop:integrator}}, the stabilizing MPC can be stated as follows\footnote{{To reduce the computational burden of the optimization problem, a control horizon smaller than the prediction horizon could be adopted \cite{rawlings2017model}.}}.

\begin{subequations} \label{eq:MPC}
	\begin{align}
		\min_{v_{0|k}, ..., v_{N_p-1|k}} &  \sum_{i=0}^{N_p} \Big[ \big\| \chi_{i|k} - \bar{\chi} \big\|_Q^2 + \big\| \zeta_{i|k} - \bar{\zeta} \big\|_R^2 \Big]\label{eq:MPC:cost} \\
		\text{s.t.} \quad & \forall i \in \{ 0, ..., N_p-1 \} \nonumber \\
		&  \chi_{0|k} = \chi_{k} \label{eq:MPC:x0} \\
		& \chi_{i+1|k} = f_a(\chi_{i|k}, v_{i|k}, \bar{y}) \label{eq:MPC:dynamics} \\
		& \zeta_{i|k} = g_a(\chi_{i|k}) \label{eq:MPC:output} \\
		& E_u \zeta_{i|k} \in \mathcal{U} \label{eq:MPC:actuator} \\
		& \chi_{N_p|k} = \bar{\chi} \label{eq:MPC:terminal}
	\end{align}
\end{subequations}

The adopted cost function \eqref{eq:MPC:cost} penalizes the deviation of the augmented system's state and output from their value at equilibrium.
Note that the target equilibrium for the integrator state $\xi$ is $\bar{u}$, as the integral action is assumed to provide the input's equilibrium at steady state.
The equilibrium value of the derivator state, $\bar{v}$, is arbitrary but constant, so that $\bar{\gamma}=0$, i.e. the contribution of the derivative action converges to zero.
The weight matrices are defined as $Q = \text{diag}(Q_x, Q_\xi, Q_\theta)$ and $R = \text{diag}(R_e, R_u)$, where $\text{diag}(\cdot)$ indicates the block-diagonal operator. 
The weights $R_e$ and $R_u$ penalize the output error and control effort, respectively, while $Q$ penalizes the distance of the augmented state from its equilibrium value. 
{Note that the deviation of $\theta$ from its arbitrary equilibrium is only penalized for numerical reasons.  It is hence advisable to select $Q_\theta \ll Q_x, Q_\xi$.}

The augmented model $\Sigma_a$ is used as predictive model, see \eqref{eq:MPC:dynamics} and \eqref{eq:MPC:output}, and it is initialized in the known state $\chi_k$, see \eqref{eq:MPC:x0}.
The inputs saturation constraints are also embedded via constraint \eqref{eq:MPC:actuator}, where $E_u$ is a selection matrix that extracts $u_{i|k}$ out of $\zeta_{i|k} = [y_{i|k}^\prime, u_{i|k}^\prime]^\prime$.
Lastly, {as customary in MPC}, the terminal constraint \eqref{eq:MPC:terminal} is imposed.

According to the Receding Horizon principle, at time $k$ the optimization problem \eqref{eq:MPC} is solved, retrieving the optimal control sequence $v^*_{0|k}, ..., v^*_{N_p-1|k}$, and only the first control move, i.e. $v^*_{0|k}$, is applied.
At the successive time step the procedure is repeated, based on the measured state $\chi_{k+1}$. \smallskip

\begin{remark}
The MPC law formulated in \eqref{eq:MPC} is a standard MPC with terminal constraint. 
Hence, its nominal recursive feasibility and closed-loop stability can be guaranteed \cite{rawlings2017model}.
\end{remark}

\section{Numerical Example} \label{sec:example}
\subsection{Benchmark system description}\label{sec:example:plant}
\begin{figure}
	\centering
	\includegraphics[width=\ifbool{extended}{0.9}{0.75 } \linewidth]{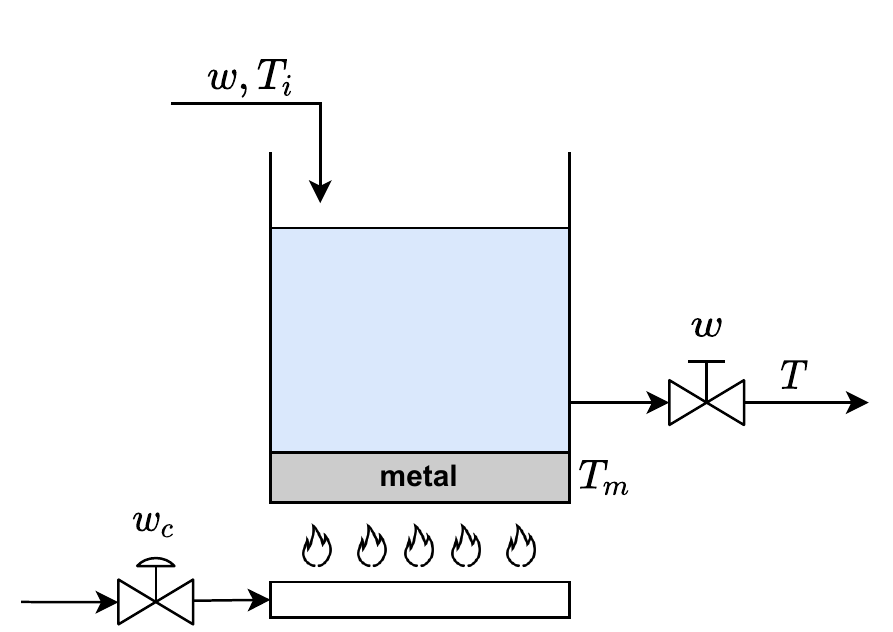}
	\caption{Water-heating system illustration}
	\label{fig:plant}
\end{figure}

The proposed control architecture has been tested on the water-heating benchmark system depicted in Figure \ref{fig:plant}.
The objective of this system is to control the temperature of the water in a reservoir so as to provide the users with the required flow of water at the desired temperature. 
Specifically, the water is heated through a metal plate placed under the tank, which is heated by means of a gas burner.

The water demand $w$, expressed in $kg/s$, represent a disturbance.
For simplicity, it is assumed that the water flow rate at the inlet matches the demand, so that the level dynamics are neglected.
We indicate by $T_i$ the temperature of the water at the inlet, and by $T$ the temperature of the water served to the users.
Both temperatures are expressed in $K$, and the water temperature is assumed to be uniform throughout the tank.
The water is heated by the metal plate, having temperature $T_m$, which is radiated by the flames resulting from the combustion of the gas, whose flow rate is denoted by $w_c$.
\begin{extendedonly}
	{Assuming the absence of heat loss, and that the flame heat is exchanged only via radiation,} the following model of the system can hence be formulated: 
	\begin{equation}\label{eq:example:plant}
	\scalemath{0.9}{
	\mathcal{P}: \,
	\begin{dcases}
 	    \dot{T} =\frac{1}{\rho_w A_t z_w}\left [ w \left ( T_{i} - T \right )+\frac{k_{lm} A_t}{c_w}\left ( T_{m }- T \right ) \right ] \\
 	   \dot{T}_m = \frac{1}{M_{m} c_{m}}\left [ -k_{lm} A_t \left ( T_{m} - T \right )+\sigma k_{f} w_{c} \left ( T_{f}^{4}-T_{m}^{4} \right ) \right ]
	\end{dcases}}.
	\end{equation}
	
	This model has one controllable input $u = [ w_c ]$, one output $y_p = [ T ]$, and two states $x_p = [ T, T_m ]^\prime$.
	Moreover, system \eqref{eq:example:plant} is affected by two disturbances, $d_p = [w, T_i]^\prime$, whose nominal values are reported, alongside the other parameters of the model, in Table \ref{tab:system_parameters}.
	The gas flow rate $w_c$ is also subject to saturation, i.e.
	\begin{equation}
 	   w_c \in [0.05, 0.18].
	\end{equation}
	
	\begin{table}[t]
		\centering
		\caption{Benchmark system parameters}
		\label{tab:system_parameters}
		\resizebox{\columnwidth}{!}{
		\begin{tabular}{clcc}
		\toprule
		Parameter  & Description & Value  & Units \\ \midrule
		$A_t$ & Tank's cross-section &$\frac{\pi}{4}$ & $m^2$ \\
		$\rho_{w}$ & Water's density & $997.8 $ & $\frac{kg}{m^3}$\\
		$c_w$ & Water's specific heat & $4180.0$ & $\frac{J}{kg \cdot K}$\\
		$M_{m}$ & Metal plate's mass & $617.32$ & $kg$\\
		$c_m$ & Metal's specific heat & $481.0$ & $\frac{J}{kg \cdot K}$\\
		$\sigma$ & Radiation coefficient & $5.67\times 10^{-8}$ & $\frac{W}{m^2 \cdot K^4}$ \\
		$k_{lm}$ & Heat exchange coefficient & $3326.4$ & $\frac{kg}{s^3 \cdot K}$ \\
		$T_f$ & Flame's temperature & $1200$ & $K$\\
		$k_f$ & Heat exchange coefficient & $8.0$ & $\frac{m^2 \cdot s}{kg}$\\
		$z_w$ & Water level & $2.0$ & $m$ \\
		$\bar{w}$ & Nominal water demand & $1.0$ & $\frac{kg}{s}$ \\
		$\bar{T}_i$ & Nominal inlet water temperature & $298$ & $K$ \\
		\bottomrule
		\end{tabular}}
	\end{table}
\end{extendedonly}
\begin{shortonly}
	Hence, the model has one controllable input $u = [ w_c ]$, one output $y_p = [ T ]$, and two states $x_p = [ T, T_m ]^\prime$, and it is affected by two disturbances, $d_p = [w, T_i]^\prime$.
	The equations of the process and the physical parameters are reported in the extended version of the paper \cite{bonassi2022extended}.
\end{shortonly}

A simulator of the benchmark system has been implemented in Simulink, so as to collect the training data and to test the proposed control architecture.

\subsection{NNARX model training}
In order to collect the data used to train the NNARX model of the plant \ifbool{extended}{\eqref{eq:example:plant}}{}, the simulator has been forced with a Multilevel Pseudo-Random Signal (MPRS), so as to properly excite the system.
The input-output data has been recorded with sampling time $\tau_s = 120 \, s$, for a total of $2500$ time steps.
According to the Truncated Back-Propagation Through Time (TBPTT) principle, $N_{t} = 120$ random subsequences of length $T_s = 400$, denoted by $(\bm{u}_{T_s}^{\{i\}}, \bm{y}_{p, T_s}^{\{i\}})$, with $i \in \mathcal{I}_t = \{ 1, ..., N_t \}$, have been extracted from the experiment data.
Other two shorter experiments have been performed, from which $N_v=30$ and $N_f=1$ subsequences have been extracted as validation and independent test datasets, respectively.
These two datasets are denoted by the set of indices $\mathcal{I}_v = \{ N_t + 1, ..., N_t + N_v \}$ and $\mathcal{I}_f = \{N_t + N_v + 1, ..., N_s \}$, respectively,  where $N_s = N_t + N_v + N_f$.

The training procedure has been conducted using PyTorch 1.9.
The adopted NNARX model features a single-layer ($M=1$) FFNN regression function with $30$ neurons, and the past $N=5$ input-output data has been used as regressors.
Following the guidelines of \cite{bonassi2022survey, bonassi2021nnarx}, the model has been trained by minimizing the simulation Mean Square Error (MSE) over the training set $\mathcal{I}_t$, i.e. the MSE between the NNARX open-loop prediction (given the input sequence $\bm{u}_{T_s}^{\{i\}}$) and the measured output sequence $\bm{y}_{p, T_s}^{\{i\}}$.
The $\delta$ISS property of the NNARX model has been enforced during training by including a suitable regularization term in the loss function \cite{bonassi2021nnarx}, so that Assumption \ref{ass:delta_iss} is satisfied.
The NNARX model has been trained for $1288$ epochs, when its modeling performances on the validation set stopped improving.

\begin{extendedonly}
	\begin{figure}[t]
		\centering
		\includegraphics[width=\ifbool{extended}{0.95}{0.8} \linewidth]{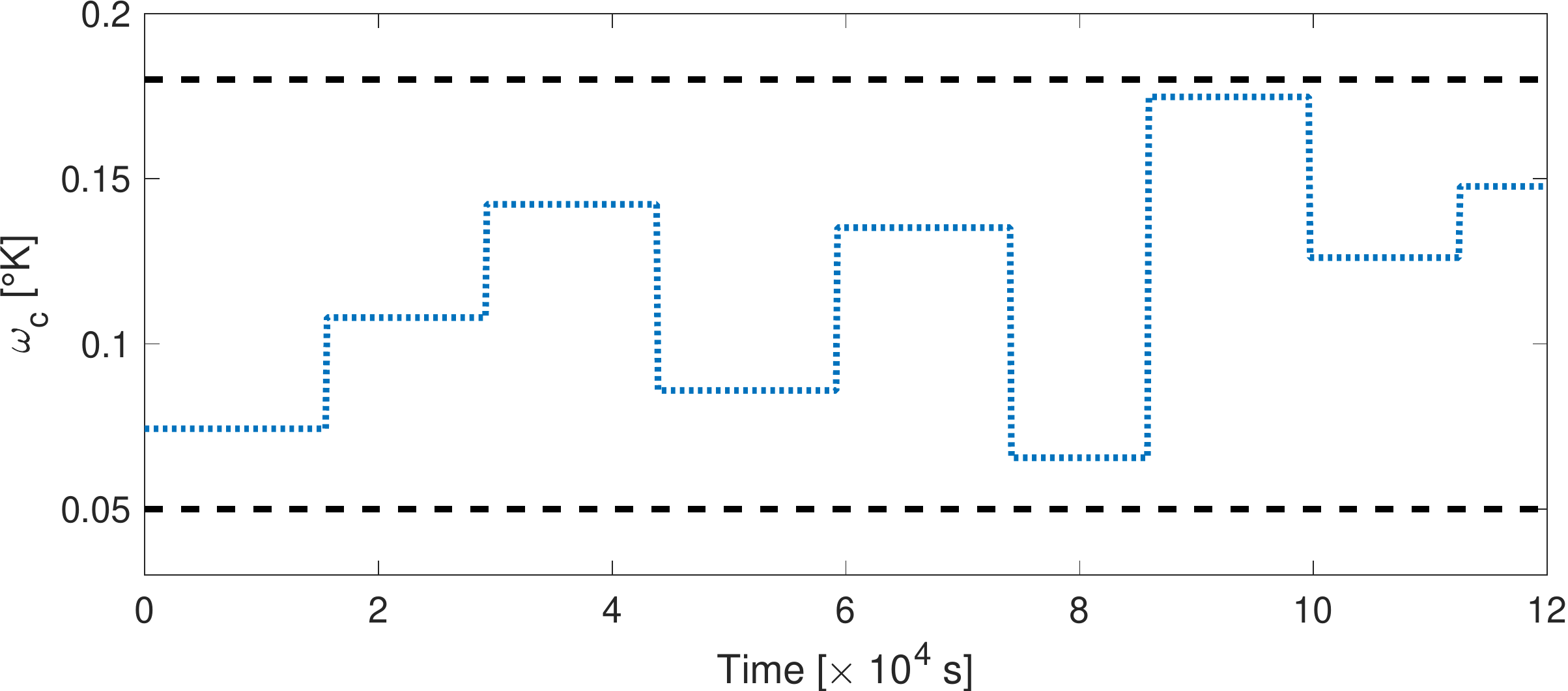}
		\caption{Input sequence used for testing the performances of the trained NNARX model.}
		\label{fig:test_input}
	\end{figure}
	Eventually, the model performances have been assessed on the independent test set. 
	In Figure \ref{fig:test_input} the test input sequence, $\bm{u}_{T_s}^{\{\iota \}}$, with $\iota \in \mathcal{I}_{f}$, is depicted, while in Figure \ref{fig:test_output} the corresponding open-loop simulation of the NNARX model is compared to the ground truth $\bm{y}_{p, T_s}^{\{\iota\}}$.
\end{extendedonly}
\begin{shortonly}
	In Figure \ref{fig:test_output} the open-loop simulation of the NNARX model corresponding to the input sequence  $\bm{u}_{T_s}^{\{\iota \}}$, with $\iota \in \mathcal{I}_{f}$, is compared to the ground truth $\bm{y}_{p, T_s}^{\{\iota\}}$.
\end{shortonly}
The performance have been quantified using the FIT [\%] index, defined as
\begin{equation}
	\text{FIT} = 100 \left( 1 -  \frac{\sum_{k=0}^{T_s} \| y_k^{\{ \iota \}} - y_{p, k}^{\{ \iota \}}\|_2}{\sum_{k=0}^{T_s}\| y_{p, k}^{\{ \iota \}} - y_{avg} \|_2} \right),
\end{equation}
where $y_k^{\{ \iota \}}$ indicates the output of the NNARX model \eqref{eq:nnarx:compact} fed by the input sequence $\bm{u}_{k}^{\{ \iota \}}$ and initialized in a random initial state, and $y_{avg}$ is the average of $y_{p,k}^{\{ \iota\}}$ over time.
The FIT scored by the trained model is $92.8 \%$, which indicates fair modeling performances of the NNARX model.

\begin{figure}[t]
	\centering
	\includegraphics[width=\ifbool{extended}{0.95}{0.8} \linewidth]{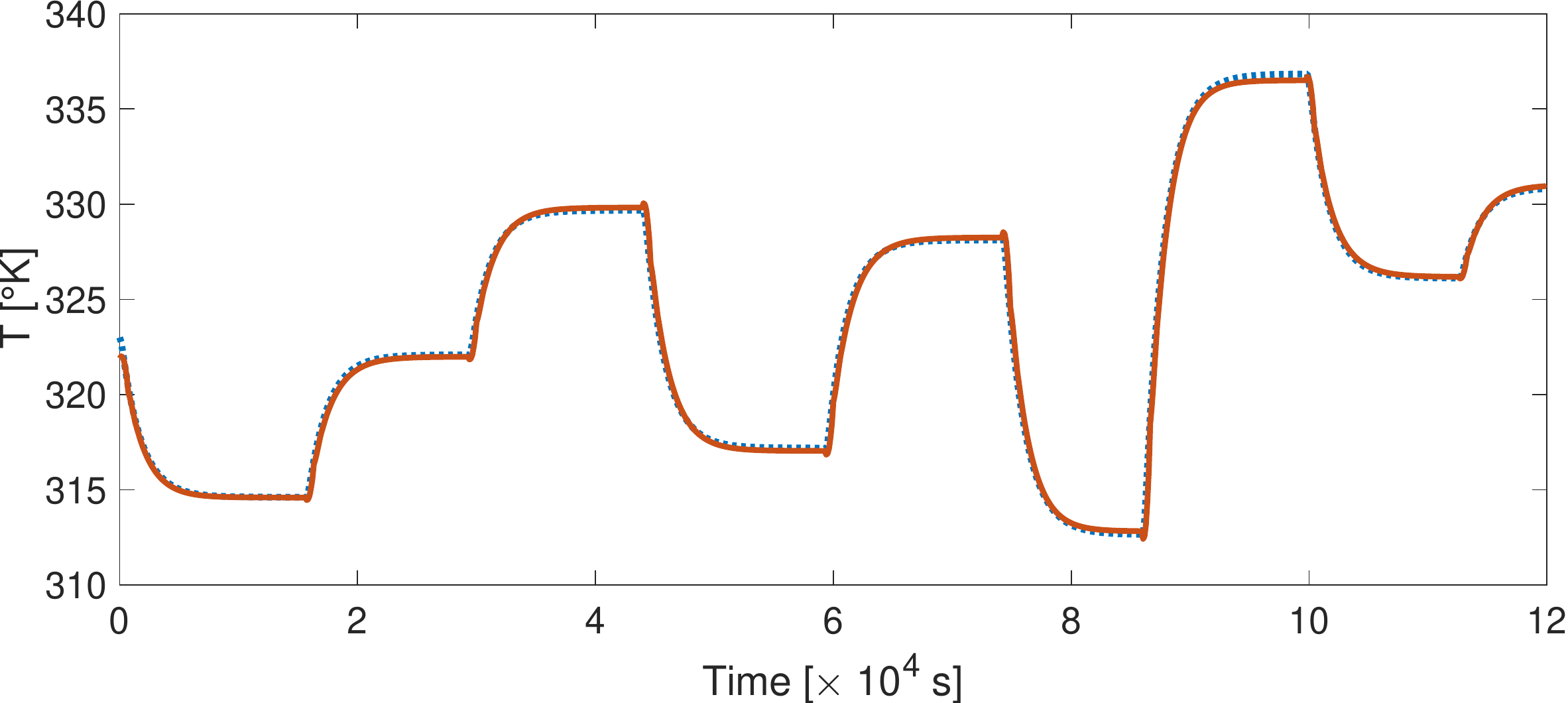}
	\caption{Modeling performances of the trained NNARX model: open-loop prediction (red line) versus ground truth (blue dotted line).}
	\label{fig:test_output}
\end{figure}

Having trained a $\delta$ISS NNARX model of the system, the control architecture proposed in Section \ref{sec:control} can be tested.

\begin{figure}[t]
	\centering
	\includegraphics[width=\ifbool{extended}{0.95}{0.8} \linewidth]{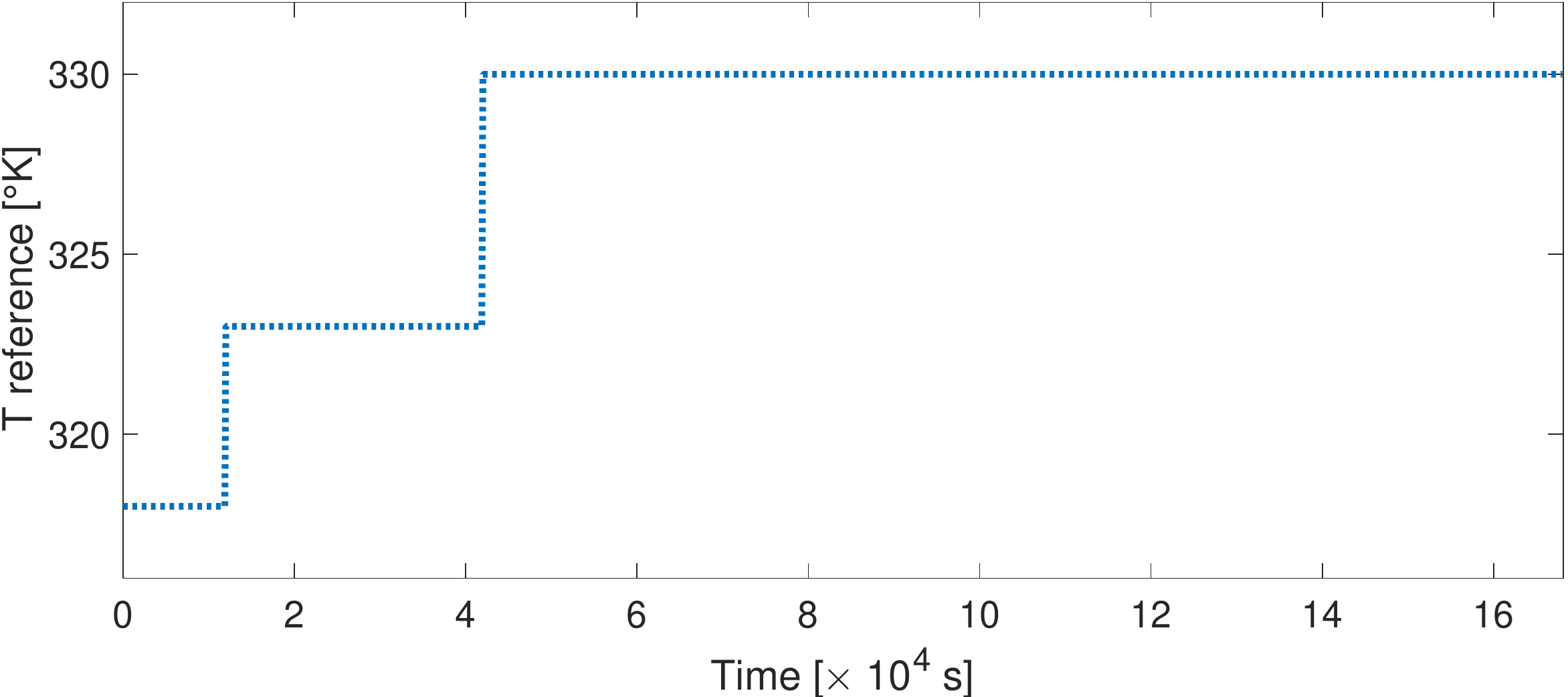}
	\vspace{-2mm}
	\caption{Piecewise-constant output reference trajectory.}
	\label{fig:reference}
	\vspace{2mm}
	\centering
	\includegraphics[width=\ifbool{extended}{0.975}{0.84}\linewidth]{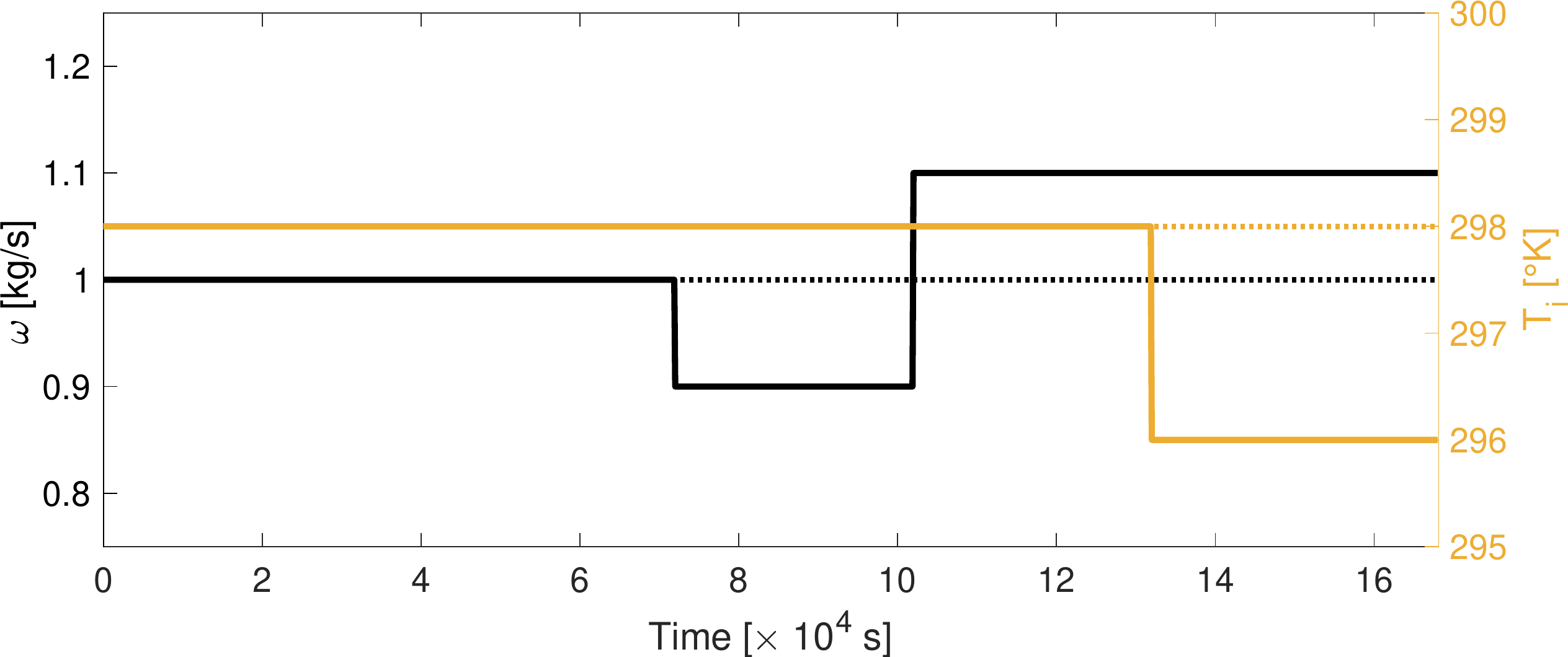}
	\vspace{-2mm}
	\caption{Disturbance applied to the system. Water demand $w$ (black continuous line) compared to its nominal value (black dotted line), and inlet water temperature $T_i$ (yellow continuous line) compared to its nominal value (yellow dotted line).}
	\label{fig:disturbances}
\end{figure}

\subsection{Control synthesis}
The proposed control architecture has been implemented, with the primary goal of tracking piecewise-constant water temperature references, while asymptotically rejecting possible (unmeasured) disturbances associated to variations of the water demand $w$ or the water temperature at the inlet $T_i$.
To assess the offset-free tracking capabilities and the robustness of the proposed architecture, the temperature reference depicted in Figure \ref{fig:reference} and the disturbances illustrated in Figure \ref{fig:disturbances} have been considered.

It should be noted that, since the output reference is piece-wise constant, at every change of the setpoint $\bar{y}$, the nominal equilibrium triplet $(\bar{x}, \bar{u}, \bar{y})$ of the corresponding nominal equilibrium of the augmented system $(\bar{\chi}, \bar{v}, \bar{\zeta})$ needs to be computed by solving \eqref{eq:control:equilibrium_def}.
Moreover, if the new setpoint is not in the neighborhood of the previous equilibrium, Assumption \ref{ass:linearized} should also be verified, and the integral action's gain $\mu$ should be recomputed according to Corollary \ref{prop:integrator}.
The prediction horizon of the adopted MPC law is $N_p=50$, while the weights are chosen as $R_e = 10$, $R_u = 0.1$, $Q_x = \text{diag}(R, R, R, R, R)$, $Q_\xi = 1$, and $Q_\theta = 10^{-5}$.

Concerning the tuning of the integral action gain $\mu$, Corollary \ref{prop:integrator} has been verified numerically. 
For the three setpoints considered in Figure \ref{fig:reference}, any gain $\mu \in (0, 0.251)$ results in a locally asymptotically stable augmented system. Specifically, $\mu = 0.14$ has been chosen.
 \smallskip

\noindent \emph{Alternative approach for comparison} 

In order to evaluate the performances of the proposed control architecture, the popular offset-free MPC strategy described in \cite{morari2012nonlinear} has also been implemented. 
In brief, this control strategy, henceforth named Disturbance Estimation Based MPC (DEB-MPC), requires to augment the NNARX model with a fictitious matched disturbance {on the input}, for the estimation of which a Moving Horizon Estimator (MHE) is designed.
Then, a standard state-feedback nonlinear MPC law is designed to stabilize the augmented model, featuring a prediction horizon $N_p=50$, and weights in line with those used for the proposed control architecture.
\smallskip

\noindent \emph{Closed-loop results}

The closed-loop output tracking performances achieved by the proposed approach are compared to those of DEB-MPC in Figure \ref{fig:closedloop_output}.
It is apparent that, while initially the DEB-MPC scheme is able to compensate the plant-model mismatch thanks to a reliable estimate of the fictitious matched disturbance, after the instant $t=7 \cdot 10^4 s$ -- when disturbances occur, see Figure \ref{fig:disturbances} -- the tracking performances of such control scheme are lost.
In contrast, the proposed control architecture is  able to attain zero tracking error, even in the presence of the severe disturbances that affect the system.
In Figure \ref{fig:closedloop_input}, the control variable requested by the two schemes is compared.
In both schemes the input constraints are fulfilled, although it can be observed that the control action issued by DEB-MPC is less moderate, mainly due to the transients of the disturbance estimator. 
{It should be noted that in DEB-MPC, the choice of the disturbance model is crucial to achieve satisfactory performance. 
Recent works proposing alternative disturbance estimation-based strategies, see \cite{tatjewski2020algorithms}, will thus be object of future investigations. 
}

\begin{figure}[t]
	\centering
	\includegraphics[width=\ifbool{extended}{0.95}{0.8} \linewidth]{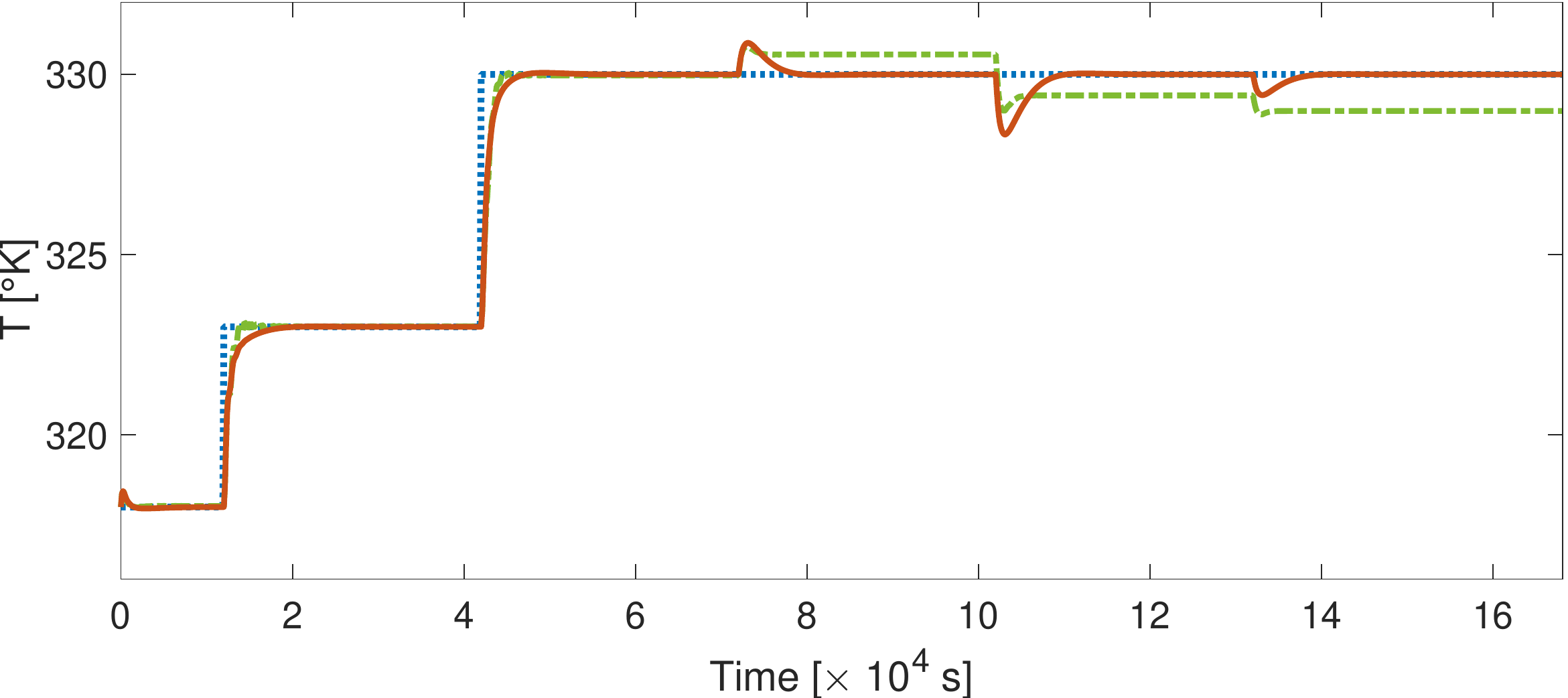}
		\vspace{-2mm}
	\caption{Closed-loop output tracking performances of the proposed approach (red continuous line) compared to that of the DEB-MPC (green dashed-dotted line). The setpoint is represented by the blue dotted line.
	}
	\label{fig:closedloop_output}
	\vspace{3mm}
	\centering
	\includegraphics[width=\ifbool{extended}{0.95}{0.8} \linewidth]{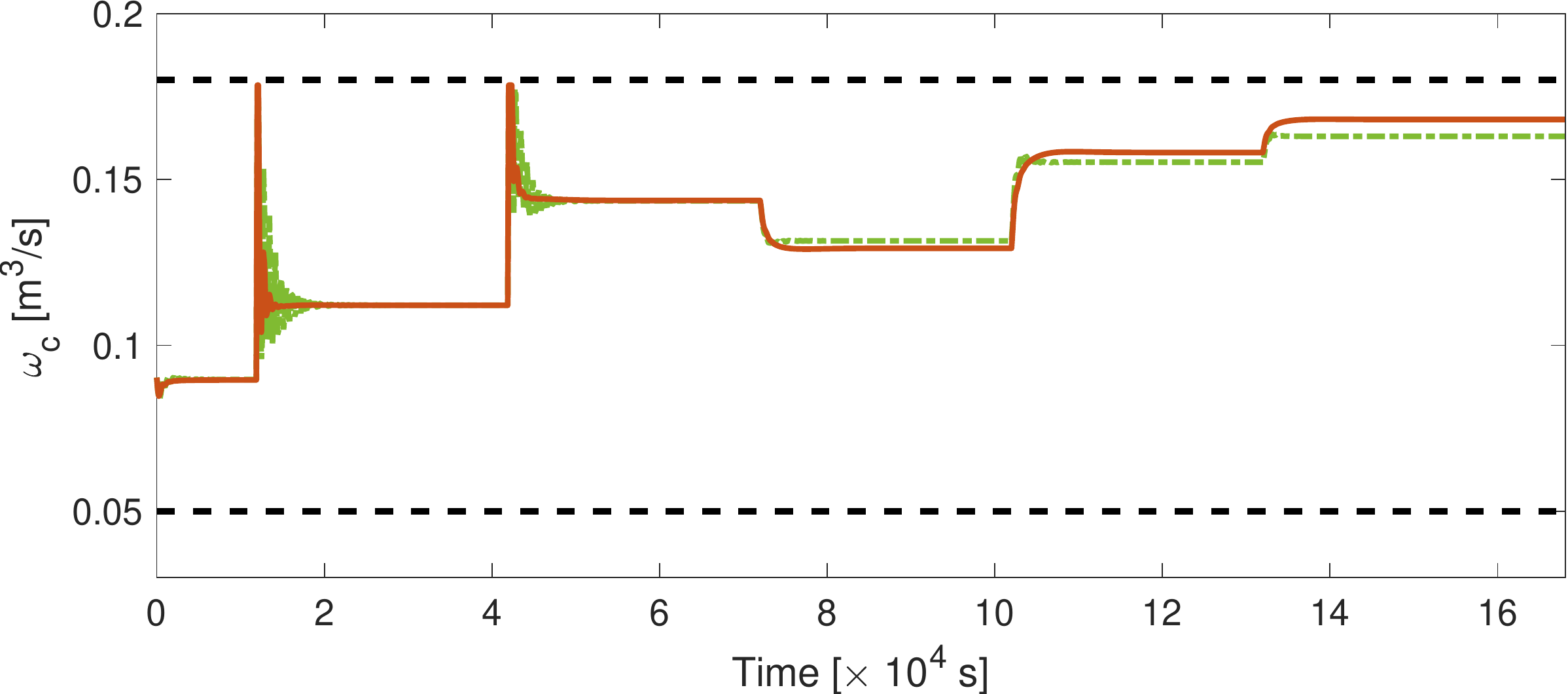}
		\vspace{-2mm}
	\caption{Control action of the proposed approach (red continuous line) compared to that of the DEB-MPC (green dashed-dotted line). }
	\label{fig:closedloop_input}
\end{figure}

\section{Conclusions} \label{sec:conclusions}
In this paper, a nonlinear Model Predictive Control (MPC) strategy is proposed to achieve offset-free tracking of constant references for system learned by Neural NARX (NNARX) models.
To this end, we proposed to augment the NNARX model with the integrator of the output tracking error and with a derivative action, and then to design a stabilizing MPC law for such augmented model.
The proposed control scheme attains nominal closed-loop stability and offset-free tracking capabilities.
The control law was tested on a water heating benchmark system, demonstrating satisfactory closed-loop performance and a good degree of robustness to the disturbances affecting the plant.

\section*{Acknowledgements}
\vspace{1mm}
\begin{minipage}[l]{0.075\textwidth}
	\includegraphics[width=\textwidth]{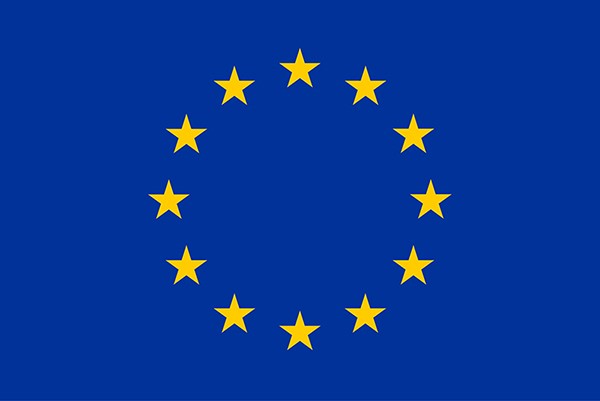}
	\label{fig:euflag}
\end{minipage}
\hspace{3mm}
\begin{minipage}[right]{0.34\textwidth}
	This project has received funding from the European Union’s Horizon 2020 research and innovation programme under the Marie Skłodowska-Curie grant agreement No. 953348
\end{minipage}
\vspace{2mm}

\bibliographystyle{IEEEtran}
\bibliography{narx_control}

 \begin{extendedonly}
 	\appendix
	\section{Proofs}
\subsection{Proof of Proposition \ref{prop:asymptotic_stability}} \label{appendix:proof_stability}
Let $\delta x_k$ and $\delta u_k$ be the displacement from the equilibrium point $(\bar{x}, \bar{u})$, i.e. $x_k = \bar{x} + \delta x_k$ and $u_k = \bar{u} + \delta u_k$.
The nonlinear system \eqref{eq:nnarx:compact} can be rewritten as
\begin{equation} \label{eq:proof:delta_system_full}
	\delta x_{k+1} + \bar{x} = f(\bar{x} + \delta x_k, \bar{u} + \delta u_k).
\end{equation}
Since the goal is to analyze the asymptotic stability of the linearized system, for simplicity it is assumed that $\delta u_k = 0$. 
It is worth noticing, however, that this proof could be easily extended to consider $\delta u_k \neq 0$, at the price of more involved computations.
Under this simplification, \eqref{eq:proof:delta_system} reads
\begin{equation} \label{eq:proof:delta_system}
	\delta x_{k+1} + \bar{x} = f(\bar{x} + \delta x_k, \bar{u}).
\end{equation}

System \eqref{eq:proof:delta_system} can be recast as its linearization plus the linerization error $\varepsilon$
\begin{equation} 
	\delta x_{k+1} = A_\delta \delta x_k + \varepsilon(\delta x_k),
\end{equation}
where
\begin{equation}
	A_\delta = \left. \frac{\partial f}{\partial x_k} \right\lvert_{\bar{x}, \bar{u}}.
\end{equation}
The goal is to show that the linear system
\begin{equation}\label{eq:proof:linear}
    \delta x_{k+1} = A_\delta \delta x_k
\end{equation}
is asymptotically stable.
Along the lines of \cite{khalil2002nonlinear}, the linearization error is first bounded as follows.

Consider the $i$-th state component, with $i \in \{ 1, ..., n \}$. In light of the Mean Value Theorem there exists $\tilde{x}$ between $\bar{x}$ and $\bar{x} + \delta x_k$ such that
\begin{equation}
\scalemath{0.9}{
    \begin{aligned}
    &f_i(\bar{x} + \delta x_k, \bar{u}) - f_i(\bar{x}, \bar{u}) = \left. \frac{\partial f_i(x, \bar{u})}{\partial x} \right\lvert_{\tilde{x}, \bar{u}}  \delta x_k \\
    & \,\, = \left.\frac{\partial f_i(x, \bar{u})}{\partial x} \right\lvert_{\bar{x}, \bar{u}} \!\! \delta x_k + \bigg[ \left. \frac{\partial f_i(x, \bar{u})}{\partial x} \right\lvert_{\tilde{x}, \bar{u}} \! - \! \left.\frac{\partial f_i(x, \bar{u})}{\partial x} \right\lvert_{\bar{x}, \bar{u}} \bigg] \delta x_k \\
    & \,\, = A_{\delta i} \delta x_k + \tilde{\varepsilon}_i(\delta x_k) \delta x_k
    \end{aligned}}
\end{equation}
Under the customary assumption that the gradient of $f(x, \bar{u})$ with respect to $x$ is Lipschitz continuous with Lipschitz constant $L_1$, it holds that
\begin{equation}
    \begin{aligned}
        \| \tilde{\varepsilon}_i(\delta x_k) \|^2_2 &\leq \left\| \left. \frac{\partial f_i(x, \bar{u})}{\partial x} \right\lvert_{\tilde{x}, \bar{u}} \! - \! \left.\frac{\partial f_i(x, \bar{u})}{\partial x} \right\lvert_{\bar{x}, \bar{u}} \right\|^2_2 \\
        & \leq L_1^2 \| \tilde{x} - \bar{x} \|_2^2 \leq L_1^2 \| \delta x_k \|_2^2.
    \end{aligned}
\end{equation}
Hence, being $\varepsilon(\delta x_k) = \tilde{\varepsilon}(\delta x_k) \delta x_k$,
 the linearization error can be bounded as
\begin{equation} \label{eq:proof:linearization_err_bound}
    \begin{aligned}
        \| \varepsilon(\delta x_k) \|_2 &\leq \| \delta x_k \|_2 \, \sqrt{\sum_{i=1}^n \big\| \tilde{\varepsilon}_i(\delta x_k ) \big\|_2^2 \, \|\delta x_k \|^2_2}  \\
        &\leq L_\varepsilon \| \delta x_k \|_2^2,
    \end{aligned}
\end{equation}
where $L_\varepsilon = L_1 \sqrt{n}$. 

At this stage, let us recall that the $\delta$ISS property implies the Global Asymptotic Stability (GAS) of any equilibrium. 
Indeed, recalling that $\delta u_k = 0$, from \eqref{eq:deltaiss:definition} it follows that
\begin{equation*}
    \| x_k - \bar{x} \|_2 \leq \beta(\| x_0 - \bar{x} \|_2, k).
\end{equation*}
Moreover, in light of the assumption on $\beta$, the exponential GAS property of any equilixbrium can be shown, since
\begin{equation} \label{eq:proof:exponential_gas}
    \| \delta x_k \|_2 \leq \rho  \| \delta x_0 \|_2 \lambda^k.
\end{equation}
This allows to invoke Theorem 5.8 of \cite{bof2018lyapunov} which, under the assumption of exponential GAS, guarantees the existence of a quadratic Lyapunov function $V(\delta x)$ for the nonlinear system \eqref{eq:proof:delta_system}.
That is, there exist positive constants $c_1$, $c_2$, $c_3$, $c_4$, such that
\begin{subequations}
    \begin{gather}
    c_1 \| \delta x_k \|_2^2 \leq V(\delta x_k) \leq c_2 \| \delta x_k \|_2^2, \label{eq:proof:lyapunov:bounds}\\
    V(A_\delta \delta x_{k} + \varepsilon(\delta x_k)) - V(\delta x_k) \leq -c_3 \| \delta x_k \|_2^2 \label{eq:proof:lyapunov:decreasing} \\
    \left\| \frac{\partial V(A_\delta \delta x_k + \varepsilon(\delta x_k))}{\partial \varepsilon} \right\|_2 \leq c_4 \| \delta x_k \|_2. \label{eq:proof:lyapunov:lipschitz}
    \end{gather}
\end{subequations}
The goal is to show that $V(\delta x_k)$ is also a Lyapunov function for the linear system \eqref{eq:proof:linear}.
To this end, let us add and subtract $V(A_\delta \delta x_k)$ from the left-hand side of \eqref{eq:proof:lyapunov:decreasing}, leading to
\begin{equation}\label{eq:proof:lyapunov:int1}
\begin{aligned}
    & V(A_\delta \delta x_{k}) - V(\delta x_k)  + \big[ V(A_\delta \delta x_{k} + \varepsilon(\delta x_k)) - V(A_\delta \delta x_{k}) \big] \\
    & \quad \leq -c_3 \| \delta x_k \|_2^2 
\end{aligned}
\end{equation}
In light of \eqref{eq:proof:lyapunov:lipschitz} and \eqref{eq:proof:linearization_err_bound}, $V(A_\delta \delta x_{k} + \varepsilon(\delta x_k)) - V(A_\delta \delta x_{k})$ can be bounded as
\begin{equation} \label{eq:proof:lyapunov:int2}
\begin{aligned}
    & \big\| V(A_\delta \delta x_{k} + \varepsilon(\delta x_k)) - V(A_\delta \delta x_{k}) \big\|_2  \\
    & \qquad \leq \left\| \frac{\partial V(A_\delta \delta x_k + \varepsilon(\delta x_k))}{\partial \varepsilon} \right\|_2 \, \| \delta x_k \|_2 \\
    & \qquad \leq c_4 L_\varepsilon \| \delta x_k \|_2^3.
\end{aligned}
\end{equation}

Owing to the bound \eqref{eq:proof:lyapunov:int2} and to the exponential GAS \eqref{eq:proof:exponential_gas}, recalling that $\lambda \in (0, 1)$, from \eqref{eq:proof:lyapunov:int1} it holds that
\begin{equation}
\begin{aligned}
    V(A_\delta \delta x_{k}) - V(\delta x_k) &\leq - c_3 \| \delta x_k \|_2^2 - c_4 L_\varepsilon \| \delta x_k \|^3_2  \\
    &\leq - c_3 \| \delta x_k \|_2^2 - \rho c_4 L_\varepsilon \| \delta x_0 \|_2 \| \delta x_k \|_2^2  \\
    &\leq - \big(c_3 - \rho c_4 L_\varepsilon \| \delta x_0 \|_2 \big) \| \delta x_k \|_2^2.
\end{aligned}
\end{equation}
Hence, there exist constants $c_5>0$ and $r_0 > 0$ such that, $\forall \delta x_0 \in \{ \delta x_0 : \| \delta x_0 \|_2 \leq r_0 \}$, 
\begin{equation*}
    V(A_\delta \delta x_{k}) - V(\delta x_k) \leq - c_5 \| \delta x_k \|_2^2.
\end{equation*}
Hence, the asymptotic stability of the linear system \eqref{eq:proof:linear} is proven by using $V(\delta x_k)$ as Lyapunov function, which implies that $A_\delta$ is Schur stable. $\hfill\blacksquare$ 
 \end{extendedonly}

\end{document}